\documentclass[a4paper,11pt]{article}
\usepackage{fullpage}

\usepackage[utf8]{inputenc}
\usepackage[colorlinks,citecolor=Blue,linkcolor=Red,urlcolor=Blue]{hyperref}

\PassOptionsToPackage{compress}{natbib}
\usepackage[numbers]{natbib}
\usepackage{libpapers}
\usepackage{libgraphics}

\usepackage{pgfplots}
\usepackage{pgfplotstable}

\definecolor{darkgreen}{rgb}{0.0, 0.5, 0.0}
\definecolor{darkyellow}{rgb}{0.99, 0.73, 0.01}

\pgfplotsset{compat=1.18}

\usepackage{algorithm}
\usepackage{algpseudocode}
\usepackage{booktabs}
\usepackage{nicematrix}

\newcommand{\ftH}{\mathbf{H}}

\newcommand{\sigw}{\Tilde{w}}
\newcommand{\sigv}{\Tilde{v}}
\newcommand{\sigx}{\Tilde{x}}
\newcommand{\sigr}{\Tilde{r}}

\title{Generalized Quantum Signal Processing and Non-Linear Fourier Transform are equivalent}
\author{Lorenzo Laneve}
\affil{\small \emph{Faculty of Informatics --- Università della Svizzera Italiana, 6900 Lugano, Switzerland}}
\date{}

\begin{document}

\maketitle

\begin{abstract}
    Quantum signal processing (QSP) and quantum singular value transformation (QSVT) are powerful techniques for the development of quantum procedures. They allow to derive circuits preparing desired polynomial transformations. Recent research [Alexis \emph{et al.}~2024] showed that \emph{Non-Linear Fourier Analysis} (NLFA) can be employed to numerically compute a QSP protocol, with provable stability. In this work we extend their result, showing that GQSP and the Non-Linear Fourier Transform over $SU(2)$ are the same object. This statement --- proven by a simple argument --- has a bunch of consequences: first, the Riemann-Hilbert-Weiss algorithm can be turned, with little modifications and no penalty in complexity, into a unified, provably stable algorithm for the computation of phase factors in \emph{any} QSP variant, including GQSP. Secondly, we derive a uniqueness result for the existence of GQSP phase factors based on the bijectivity of the Non-Linear Fourier Transform. Furthermore, NLFA provides a complete theory of \emph{infinite generalized quantum signal processing}, which characterizes the class of functions approximable by GQSP protocols.
\end{abstract}

\section{Introduction}
Quantum signal processing (QSP) is a technique that draw significant attention in the past years~\cite{lowHamiltonianSimulationUniform2017,lowHamiltonianSimulationQubitization2019}. In this setting, one is able to construct polynomials of a certain variable $z$ (the \emph{signal}) encoded in a \emph{signal operator} $W(z)$ by intertwining calls to said operator with single-qubit operations $A_k$, called \emph{processing operators}. By instantiating $W(z)$ as some big matrix $M$ acting on a larger subsystem, and invoking standard matrix decomposition theorems, one can transpose the role of $z$ to the eigenvalues or the singular values of $M$. These linear-algebraic techniques go by the name of \emph{quantum eigenvalue transformation} (QET)~\cite{lowOptimalHamiltonianSimulation2017,dongGroundStatePreparationEnergy2022} and \emph{quantum singular-value transformation} (QSVT)~\cite{gilyenQuantumSingularValue2019,tangCSGuideQuantum2023}.

The rich class of polynomials provided by QSP, along with the fact that the ansatz is generally carried out on a single qubit, makes QET and QSVT flexible and generally applicable tools in quantum algorithmic design: they were shown to unify --- and sometimes even improve --- many known quantum algorithms, including Hamiltonian simulation~\cite{lowHamiltonianSimulationUniform2017,lowHamiltonianSimulationQubitization2019,lowOptimalHamiltonianSimulation2017,martynEfficientFullycoherentQuantum2023}, amplitude amplification~\cite{groverFastQuantumMechanical1996,brassardQuantumAmplitudeAmplification2002}, linear system solving~\cite{harrowQuantumAlgorithmLinear2009}, factoring~\cite{gilyenQuantumSingularValue2019,martynGrandUnificationQuantum2021} and many others~\cite{mcardleQuantumStatePreparation2022,laneveRobustBlackboxQuantumstate2023,rallAmplitudeEstimationQuantum2023,sinanan-singhSingleshotQuantumSignal2023,rossiQuantumSignalProcessing2023}.

More recently, the traditional QSP theory was enlarged with what is called \emph{generalized QSP} (or GQSP)~\cite{motlaghGeneralizedQuantumSignal2024}: roughly speaking, while traditional QSP only allowed for determined subsets of $SU(2)$ as valid processing operators, GQSP allows any single-qubit operation, which corresponds to an enlarged set of achievable polynomials. This not only gave further results and improvements~\cite{berryDoublingEfficiencyHamiltonian2024a,sunderhaufGeneralizedQuantumSingular2023}, but gives a more complete treatment of single-qubit quantum signal processing, as it shows that the full class of polynomials admits a GQSP decomposition.

Usually, QSP applications consist in three steps: first one has to design a polynomial $P(z)$ approximating some desired target function. The second step is \emph{completion}, where one computes a polynomial $Q(z)$ satisfying $|P|^2 + |Q|^2 = 1$, a condition necessary for $P$ to be embedded in a quantum state. The last step is to compute numerically the processing operators $A_k$ constituting the protocol for $(P, Q)$. While the original QSP theorems are proven using an induction argument, the corresponding procedure --- the so-called \emph{layer stripping} method ---, does not guarantee numerical stability, and cannot provide reliable results for polynomials of too large degrees. An important line of research is devoted to this task, and many results were found using several strategies, such as root finding~\cite{haahProductDecompositionPeriodic2019,chaoFindingAnglesQuantum2020}, Prony's method~\cite{yingStableFactorizationPhase2022}, and iterative methods~\cite{dongEfficientPhasefactorEvaluation2021,wangEnergyLandscapeSymmetric2022,dongRobustIterativeMethod2024,dongInfiniteQuantumSignal2024}. Iterative methods are generally faster and more stable than algebraic methods, but they seem to fail in the \emph{fully-coherent regime}, i.e., when $\norm{P}_{\infty} \simeq 1$. Moreover, all these methods are developed for the traditional QSP, while numerical algorithms for GQSP are limited and not provably stable~\cite{yamamotoRobustAngleFinding2024}.

More recently, QSP was shown to have connections with the \emph{Non-Linear Fourier Transform} (or NLFT)~\cite{alexisQuantumSignalProcessing2024,taoNonlinearFourierAnalysis2012}. In particular it was shown that the inverse NLFT can be used to stably compute a QSP protocol for a desired target function, using the \emph{Riemann-Hilbert-Weiss algorithm}~\cite{alexisInfiniteQuantumSignal2024}. Moreover, NLFT provided an answer to what was called \emph{infinite quantum signal processing}~\cite{dongInfiniteQuantumSignal2024}, where the QSP protocols have infinite length, producing non-polynomial target functions.

In this work we show that the NLFT not only is useful to compute the processing operators for QSP protocols, but that GQSP and NLFT are indeed the same theory. More precisely, while traditional QSP protocols (where the $A_k$ are $X$-rotations) correspond to imaginary sequences in the pre-image of the NLFT, protocols with $Y$-rotations correspond to real sequences, and generally complex sequences yield GQSP protocols. This simple statement has many consequences: first, this gives a complete theory of \emph{infinite generalized quantum signal processing}, which can be further investigated by directly working on the NLFT. Moreover this confirms that, given a target function, computing the QSP processing operators is computationally equivalent to the computation of the inverse NLFT, regardless of the variant of QSP taken into consideration, i.e., also the latter problem is (partially) solved by the aforementioned algorithms, which can then be employed to invert the NLFT in other contexts. The third is that the Riemann-Hilbert-Weiss algorithm (presented in~\cite{alexisInfiniteQuantumSignal2024} and later improved in~\cite{niFastPhaseFactor2024}) can be turned --- with little modification and without any penalty in complexity --- into a unified algorithm for QSP factorization in any variant and, in particular, gives the first stable algorithm for GQSP, even in the fully-coherent regime.

This work starts with Section~\ref{sec:qsp-theory}, which gives a comprehensive treatment of the QSP literature, showing how the QSP zoo --- including QSVT and QET-U --- can all be reduced to a convenient version of GQSP. Section~\ref{sec:nlft-theory} will then give a brief and informal introduction to the NLFT over $SU(2)$, along with a bunch of important properties. Section~\ref{sec:gqsp-nlft-equivalence} shows the relationship between (G)QSP and NLFT, with an overview of which subsets of the NLFT domain are related to the QSP subalgebras, an extended relationship between the QSP phase factors and the inverse NLFT, as well as a uniqueness result for GQSP phase factors. Section~\ref{sec:rhw-algorithm} is devoted to review the Riemann-Hilbert-Weiss algorithm, as well as the Half Cholesky method of~\cite{niFastPhaseFactor2024}, showing how to extend them for arbitrary complex polynomials. A Python implementation of the unified Riemann-Hilbert-Weiss algorithm can be found in~\cite{rhwPythonRepo}.

\subsection{Preliminaries and notation}
For a complex matrix $A$, we use $A^*$ to denote entrywise conjugation, and $A^\dag$ for the conjugate transpose. Given a function of complex variable $f(z)$, its conjugate function is defined as $f^*(z) := \overline{f(1/\overline{z})}$, where $\overline{z}$ denotes the complex conjugate of the number $z$. We use $\D$ to denote the unit disk $\{ z \in \C : |z| \le 1 \}$, while $\T$ is used to denote the unit circle $\{ z \in \C : |z| = 1 \}$. Similarly, $\D^* := \{ 1/\overline{z} : z \in \D \}$ is the (closed) region of the complex plane outside of the unit circle. The Hardy spaces $H^2(\D)$ are the spaces of holomorphic functions on $\D$, whose norm
\begin{align*}
    \norm{f}_{H^2(\D)}^2 := \sup_{0 \le r < 1} \int_{\T} |f(rz)|^2
\end{align*}
is bounded. Roughly speaking, $f \in H^2(\D)$ can be represented by a formal power series with non-negative frequencies. Similarly, one defines the Hardy space $H^2(\D^*)$ for the holomorphic functions on $\D^*$, whose formal power series have only non-positive frequencies. We define the Pauli matrices
\begin{align*}
    X =
    \begin{bmatrix}
        0 & 1 \\
        1 & 0
    \end{bmatrix}, \ \ \ \ 
    Y =
    \begin{bmatrix}
        0 & -i \\
        i & 0
    \end{bmatrix}, \ \ \ \ 
    Z =
    \begin{bmatrix}
        1 & 0 \\
        0 & -1
    \end{bmatrix},
\end{align*}
while $H$ is the single-qubit Hadamard gate satisfying $HZH = X, HXH = Z$.
\section{General theory of single-qubit quantum signal processing}
\label{sec:qsp-theory}
We give an introduction to the theory of quantum signal processing (QSP). In particular we reduce all the variants of QSP known across the literature to a slightly modified version of Generalized QSP~\cite{motlaghGeneralizedQuantumSignal2024}, which we will then relate to the NLFT. In this setting, we have access to a single-qubit operator parameterized by a variable $z \in \T$ (the \emph{signal} operator). For example, we can take the following operator
\begin{align*}
    \sigw = \diag(z, 1) \ .
\end{align*}
We then choose a sequence $A_k \in SU(2)$ of single-qubit operations independent of $z$ (the \emph{processing} operators), and we construct a protocol as follows
\begin{align}
    A_0 \sigw A_1 \sigw \cdots \sigw A_n \ . \label{eq:analytic-qsp-ansatz}
\end{align}
If we start, say, in the $\ket{0}$ state, then by matrix multiplication one can see that the final state will be of the form
\begin{align*}
    P(z) \ket{0} + Q(z) \ket{1}
\end{align*}
where $P(z), Q(z)$ are polynomials in $z$ of degree $n$, satisfying $|P(z)|^2 + |Q(z)|^2 \equiv 1$ on $\T$. It is then natural to ask whether \emph{any} such polynomial pair $(P, Q)$ admits a decomposition of the form of~(\ref{eq:analytic-qsp-ansatz}), and the answer was proven to be positive~\cite{motlaghGeneralizedQuantumSignal2024}.

A second signal operator found in the literature~\cite{haahProductDecompositionPeriodic2019} is the following
\begin{align*}
    \sigv = \diag(z, z^{-1}) \ .
\end{align*}
Because of the presence of $z^{-1}$, the polynomials produced by a protocol using $\sigv$ will be Laurent polynomials in general, and since $z$ and $z^{-1}$ have two degrees of distance between each other, such Laurent polynomials will have definite parity, i.e., will be either even or odd in $z$. That said, similarly to the first case, any pair of normalized, definite-parity Laurent polynomials can be decomposed if one replaces $\sigw$ with $\sigv$ in (\ref{eq:analytic-qsp-ansatz}).

From now on we will call the first variant \emph{analytic} QSP, while we use \emph{Laurent} QSP to denote the second form. While Laurent QSP is the first one considered in the literature, analytic QSP is easier to visualize and has a much cleaner relationship with the NLFT, thus we will mainly focus on it, and the rest of this section is devoted to showing that every single-qubit QSP variant encountered in the literature can be reduced to the analytic case. We first show that analytic and Laurent QSP are perfectly equivalent, in the following sense:
\begin{lemma}[\cite{laneveMultivariatePolynomialsAchievable2025}]
    \label{thm:analytic-laurent-correspondence}
    If the pair $(P, Q)$ is constructed by a sequence of processing operators $A_0, \ldots, A_n$ using analytic QSP, then the pair of Laurent polynomials
    \begin{align*}
        (P', Q') := (z^{-n} P(z^2), z^{-n} Q(z^2))
    \end{align*}
    is constructed by Laurent QSP using the same processing operators.
\end{lemma}
\begin{proof}
    If $A_0 \sigw \cdots \sigw A_n \ket{0} = (P, Q)$, then by replacing $\sigw$ with $\sigw' = \diag(z^2, 1)$, we construct $(P(z^2), Q(z^2))$. The claim holds by simply noticing that $\sigv = z^{-1} \sigw'$, and collecting the $z^{-1}$.
\end{proof}
Hence, if one wants to construct some definite-parity Laurent polynomial $(P, Q)$ using Laurent QSP, this can always be converted to an analytic version $(z^{n/2} P(z^{1/2}), z^{n/2} Q(z^{1/2}))$, essentially squeezing the support from $\{ -n, \ldots, n \}$ to $\{ 0, \ldots, n \}$, and then deriving the protocol in the analytic picture. The corresponding processing operators will be valid also for the original polynomials, by replacing back $\sigw$ with $\sigv$.

Usually, either because constructions require it or for ease of implementations, the $A_k$'s are not allowed to be general $SU(2)$ unitaries. The usual choice is to consider only operators of the form
\begin{align*}
    A_k = e^{i \phi_k X}
\end{align*}
where the $\phi_k$ are usually called \emph{phase factors}. In both Laurent and analytic pictures, constraining the processing operators to be $X$-rotations gives the additional condition that $P$ must have real coefficients, and $Q$ must have purely imaginary coefficients. In this work we will refer to the case where the $A_k$'s are constrained to be $X/Y/Z$-rotations as $X/Y/Z$-constrained QSP, respectively. The following theorem sums up what said so far:
\begin{theorem}[Generalized QSP, adapted from~\cite{motlaghGeneralizedQuantumSignal2024}]
    \label{thm:generalized-qsp}
    Let $P, Q \in \C[z]$ of degree $n$. There exist phase factors $\lambda, \phi_k, \theta_k \in [-\frac{\pi}{2}, \frac{\pi}{2})$ such that:
    \begin{align*}
        e^{i\lambda Z} e^{i\phi_0 X} e^{i\theta_0 Z} \ \sigw\ e^{i\phi_1 X} e^{i\theta_1 Z} \ \sigw \cdots \sigw \ e^{i\phi_n X} e^{i\theta_n Z} \ket{0} = (P, Q)
    \end{align*}
    if and only if $|P(z)|^2 + |Q(z)|^2 \equiv 1$ on $\T$.
\end{theorem}
\noindent This is slightly different from the original GQSP theorem of~\cite{motlaghGeneralizedQuantumSignal2024}, since we are using $SU(2)$ rotations, while they use asymmetric rotations and a reflection. This version --- equivalent up to global phases ---, however, is a direct generalization of the traditional QSP ans{\"a}tze.

We remark that GQSP can be equivalently expressed in the Laurent picture, with a clear characterization given by Lemma~\ref{thm:analytic-laurent-correspondence}. In any of the two pictures, one can have an infinite sequence of $A_k$'s. This, intuitively, produces Taylor or Laurent series, showing that in the limit as $n \rightarrow \infty$, one can produce (roughly) any $H^2$ function. This goes by the name of \emph{infinite quantum signal processing}, which was studied under certain symmetry constraints~\cite{dongInfiniteQuantumSignal2024}, and then later completed under the same constraints, using Non-Linear Fourier Analysis~\cite{alexisInfiniteQuantumSignal2024}.

\begin{figure}
    \centering
    \begin{tikzpicture}[
    node distance=2.5cm and 2cm,
    arrow/.style={draw, -{Latex[length=2mm]}, thick},
    label/.style={font=\small, midway, align=center, text=black, inner sep=1pt}
]

\tikzset{
    diagnode/.style={draw, rounded corners, minimum width=3.5cm,inner ysep=0.3cm, minimum height=1cm, align=center},
}

\node[diagnode] (AnalyticQSP) at (0,0) {\small {\bfseries Analytic QSP} \\ {$\sigw = \diag(z, 1)$} \\ \cite{motlaghGeneralizedQuantumSignal2024}};
\node[diagnode] (LaurentQSP) at (0,-4) {\small {\bfseries Laurent QSP} \\ {$\sigv = \diag(z, z^{-1})$} \\ \cite{haahProductDecompositionPeriodic2019}};
\node[diagnode] (ChebyshevQSP) at (6,-4) {\small {\bfseries Chebyshev QSP} \\ $\sigx = e^{i \arccos(x) X}$ \\ \cite{lowMethodologyResonantEquiangular2016,lowOptimalHamiltonianSimulation2017}};
\node[diagnode] (ReflectionQSP) at (12,-4) {\small {\bfseries Reflection QSP} \\ $\sigr = \id - 2 \Pi_x$ \\ \cite{gilyenQuantumSingularValue2019}};
\node[diagnode] (QSVT) at (12, 0) {{\bfseries QSVT} \\ \cite{gilyenQuantumSingularValue2019}};
\node[diagnode] (QETU) at (6,0) {{\bfseries QET-U} \\ \cite{dongGroundStatePreparationEnergy2022}};

\draw[arrow,<->] (AnalyticQSP) -- (LaurentQSP) node[label,fill=white,inner sep=2pt,midway] {Lemma~\ref{thm:analytic-laurent-correspondence}};
\draw[arrow,<->] (ChebyshevQSP) -- (ReflectionQSP);

\draw[arrow] (LaurentQSP) -- (ChebyshevQSP);
\draw[arrow] (ReflectionQSP) -- (QSVT) node[label,fill=white,inner sep=2pt,midway] {CS decomposition \\ \cite{tangCSGuideQuantum2023}};
\draw[arrow] (AnalyticQSP) -- (QETU) node[label,fill=white,inner sep=2pt,midway] {Spectral\\Theorem};
\draw[arrow] (LaurentQSP) -- (QETU)
node[label,fill=white,inner sep=2pt,midway] {Spectral\\Theorem};

\end{tikzpicture}
    \caption{Overview of the QSP zoo, summarizing Section~\ref{sec:qsp-theory}. This shows that analytic QSP --- which is equal to the Non-Linear Fourier transform (Section~\ref{sec:gqsp-nlft-equivalence}) --- is complete for every possible QSP variant and thus we can restrict ourselves to this case.}
    \label{fig:qsp-variants}
\end{figure}

\subsection{Chebyshev and reflection QSP}
The reader more familiar with QSVT than QSP might find the theory above slightly different. Originally, QSP was formalized as a protocol that constructs a polynomial $P(x)$, for a real variable $x \in [-1, 1]$ given by the signal operator~\cite{lowMethodologyResonantEquiangular2016,lowOptimalHamiltonianSimulation2017}
\begin{align*}
    \sigx = \begin{bmatrix}
        x & i \sqrt{1 - x^2} \\
        i \sqrt{1 - x^2} & x
    \end{bmatrix}
\end{align*}
and constraining the signal operators to be $Z$-rotations $A_k = e^{i \phi_k Z}$. In order to obtain Chebyshev QSP from Laurent QSP (argument also shown in~\cite{haahProductDecompositionPeriodic2019, martynGrandUnificationQuantum2021}), we simply do a change of variable $x = (z + z^{-1})/2$ in a $X$-constrained Laurent QSP.
\begin{align*}
    e^{i\phi_0 X} \sigv\ e^{i\phi_1 X} \sigv \cdots \sigv\ e^{i\phi_n X} =
    \begin{bmatrix}
        P(z) & Q(z) \\
        -Q^*(z) & P^*(z)
    \end{bmatrix} \ .
\end{align*}
Since $\sigx = H \sigv H$, the above protocol is equal to the following
\begin{align*}
    H e^{i\phi_0 Z} \sigx e^{i\phi_1 Z} \sigx \cdots \sigx e^{i\phi_n Z} H = 
    H
    \begin{bmatrix}
        P'(z) & Q'(z) \\
        -Q'^*(z) & P'^*(z)
    \end{bmatrix}
    H \ .
\end{align*}
By the Hadamard transformation, we obtain:
\begin{align*}
    P' = \frac{P + P^*}{2} + \frac{Q - Q^*}{2}, \,\,\, Q' = \frac{P - P^*}{2} - \frac{Q + Q^*}{2} \ .
\end{align*}
Since we have an $X$-constrained QSP protocol, then $P \in \R[z, z^{-1}], Q \in i\R[z, z^{-1}]$, which implies that $P'(z^{-1}) = P'(z), Q'(z) = -Q'(z^{-1})$. By exchanging the variables we obtain:
\begin{align*}
    P''(x) = \sum_{k=0}^n 2 p'_k T_k(x),\,\,\, iQ''(x) \sqrt{1 - x^2} = \sum_{k=0}^n  2 i q'_k U_{k-1}(x) \sqrt{1 - x^2}
\end{align*}
where $T_k, U_k$ are the Chebyshev polynomials of first and second kind, respectively, which satisfy
\begin{align*}
    T_k(x) = \frac{z^k + z^{-k}}{2} = \cos(k\theta), \,\,\, U_{k-1}(x) \sqrt{1 - x^2} = \frac{z^k - z^{-k}}{2i} = \sin(k\theta)
\end{align*}
with $z = e^{i\theta}$ and $x = \cos \theta$. To summarize, Chebyshev QSP is of the form
\begin{align*}
    e^{i\phi_0 Z} \sigx e^{i\phi_1 Z} \sigx \cdots \sigx e^{i\phi_n Z} = 
    \begin{bmatrix}
        P''(x) & iQ''(x) \sqrt{1 - x^2} \\
        -iQ''^*(x) \sqrt{1 - x^2} & P''^*(x)
    \end{bmatrix}
\end{align*}
where definite-parity $P'', Q'' \in \C[x]$ can be constructed~\cite[Theorem 9]{martynGrandUnificationQuantum2021}. The reflection QSP is similar to the Chebyshev variant, except that the signal operator is
\begin{align*}
    \sigr =
    \begin{bmatrix}
        x & \sqrt{1 - x^2} \\
        \sqrt{1 - x^2} & -x
    \end{bmatrix} \ .
\end{align*}
The reflection variant is the heart of the quantum singular value transformation~\cite[Corollary 8]{gilyenQuantumSingularValue2019}, and leverages the fact that the signal operator is Hermitian (the QSVT theorem can be proven by showing that $\sigr$ appears in the Cosine-Sine decomposition of the block-encoding unitary~\cite{tangCSGuideQuantum2023}). The reduction to Chebyshev QSP is straightforward by replacing $$\sigr = -i e^{-i\frac{\pi}{4}Z} \sigx e^{-i\frac{\pi}{4}Z} \ .$$
We sum up all the relationships between the QSP variants in Figure~\ref{fig:qsp-variants}.
\begin{figure}
    \centering
    \definecolor{colorf}{HTML}{fff06e}
\definecolor{colora}{HTML}{83defc}
\definecolor{colorb}{HTML}{ff9494}

\def\fpos{0}
\def\apos{-2}
\def\bpos{-3.5}

\begin{tikzpicture}
    \foreach \x in {-6,-5,...,6} {
        \ifnum\x<-2
            \draw[thick] (\x,\fpos) rectangle ++(1,1);
        \else\ifnum\x>3
            \draw[thick] (\x,\fpos) rectangle ++(1,1);
        \else
            \draw[thick, fill=colorf] (\x,\fpos) rectangle ++(1,1);
        \fi\fi
        \node at (\x+0.5,{\fpos+0.5}) {$F_{\x}$};
    }
    \node at (-6.5,0.5) {$\cdots$};
    \node at (7.5,0.5) {$\cdots$};

    \node at (0.5, -0.5) {$\Downarrow$};

    \foreach \x in {-6,-5,...,6} {
        \ifnum\x<-5
            \draw[thick] (\x,\apos) rectangle ++(1,1);
        \else\ifnum\x>0
            \draw[thick] (\x,\apos) rectangle ++(1,1);
        \else
            \draw[thick, fill=colora] (\x,\apos) rectangle ++(1,1);
        \fi\fi
        \node at (\x+0.5,{\apos+0.5}) {$a_{\x}$};
    }
    \node at (-6.5,{\apos+0.5}) {$\cdots$};
    \node at (7.5,{\apos+0.5}) {$\cdots$};

    \foreach \x in {-6,-5,...,6} {
        \ifnum\x<-2
            \draw[thick] (\x,\bpos) rectangle ++(1,1);
        \else\ifnum\x>3
            \draw[thick] (\x,\bpos) rectangle ++(1,1);
        \else
            \draw[thick, fill=colorb] (\x,\bpos) rectangle ++(1,1);
        \fi\fi
        \node at (\x+0.5,{\bpos+0.5}) {$b_{\x}$};
    }
    \node at (-6.5,{\bpos+0.5}) {$\cdots$};
    \node at (7.5,{\bpos+0.5}) {$\cdots$};
\end{tikzpicture}
    \caption{Visualization of a sequence $F_k$ supported on $\{ -2, \ldots, 3 \}$ mapped by the NLFT to a pair of Laurent polynomials $(a, b)$ (the colored squares are the only non-zero terms). The degrees of $b$ follow the support of $F$, while $a$ always ends at the zero frequency (with $a_0 > 0$).}
    \label{fig:nlft-visualization}
\end{figure}

\section{The Non-Linear Fourier transform over $SU(2)$}
\label{sec:nlft-theory}

The Non-Linear Fourier transform is a mathematical tool with its own history and line of research~\cite{taoNonlinearFourierAnalysis2012,tsaiSU2NonlinearFourier2005}. We define it briefly and state some useful properties, inviting the interested reader to check~\cite{alexisQuantumSignalProcessing2024} for a more comprehensive overview.
\begin{definition}
    Let $(F_n)_{n \in \Z}$ be a complex sequence in $\ell^2(\Z)$. The \emph{Non-Linear Fourier Transform} over $SU(2)$ is defined informally as the infinite product
    \begin{align*}
        \calG_{F}(z) = 
        \prod_{n = -\infty}^{+\infty}
        \frac{1}{\sqrt{1 + |F_n|^2}}
        \begin{bmatrix}
            1 & F_n z^n \\
            -F_n^* z^{-n} & 1
        \end{bmatrix}
    \end{align*}
\end{definition}
Notice that $\calG_F$ produces a function from $\T$ to $SU(2)$. We use $(a, b)$ as a shortcut to denote the unique $SU(2)$ element whose first row is $(a, b)$, so that $\calG_F(z) = (a(z), b(z))$. We introduce some important definitions: a function $f \in H^2(\D)$ is called \emph{inner} if $|f| = 1$ almost everywhere on $\T$. Two functions $a, b$ are said to share a (non-trivial) common inner factor $f$ if $a = f a', b = f b'$ for some $a', b'$.
\begin{definition}[\cite{tsaiSU2NonlinearFourier2005, alexisQuantumSignalProcessing2024}]
    Let $\ftH_{\ge k}$ be the space of measurable functions $(a, b)$ such that:
    \begin{itemize}
        \item $a a^* + b b^* = 1$ almost everywhere on $\T$;
        \item $a \in H^2(\D^*)$ with $a(\infty) > 0$, and $b \in z^k H^2(\D)$;
        \item $a^*, b$ share no common inner factor.
    \end{itemize}
    Moreover, we denote with $\ftH = \bigcup_k \ftH_{\ge k}$.
\end{definition}
\noindent The second condition implies the Laurent series of $a$ has degrees in $(-\infty, 0]$ (with zero-frequency coefficient being positive), while the series of $b$ has degrees in $[k, +\infty)$.
\begin{theorem}[\cite{tsaiSU2NonlinearFourier2005,alexisQuantumSignalProcessing2024}]
    The NLFT is a bijection between the space $\ell^2(\Z_{\ge k})$ of one-sided square-summable sequences starting from $k$ and $\ftH_{\ge k}$.
\end{theorem}

\noindent For completeness we also report some properties of the NLFT.
\begin{theorem}[\cite{alexisQuantumSignalProcessing2024}]
    \label{thm:nlft-properties}
    Suppose we have two sequences $F, G$, and $\calG_F = (a, b)$.
    \begin{enumerate}[(i)]
        \item \textbf{Shift property:} if $G_n = F_{n-k}$, then $\calG_G = (a, z^k b)$;

        \item \textbf{Composition property:} if the support of $F$ is entirely on the left of the support of $G$, then $\calG_{F+G} = \calG_F \cdot \calG_G$;

        \item \textbf{Phase property:} if $|c| = 1$, then $\calG_{cF} = (a, cb)$;

        \item \textbf{Reflection property:} if $G_n = F_{-n}$, then $\calG_G = (a^*(z^{-1}), b(z^{-1}))$;

        \item \textbf{Conjugation property:} if $G_n = F_n^*$, then $\calG_G = (a^*(z^{-1}), b^*(z^{-1}))$.
    \end{enumerate}
\end{theorem}
\noindent Another important property is that the support of $b$ follows the one of $F$: more precisely, if $F_n$ is supported on $\{n_1, \ldots, n_2\}$, then $b(z)$ will be a polynomial of the form $$b(z) = \sum_{k = n_1}^{n_2} b_k z^k \ .$$
$a$ will thus be a polynomial with the same support, but shifted so that its last non-zero coefficient is positive and on the zero frequency (see Figure~\ref{fig:nlft-visualization}).
\section{Equivalence of GQSP and NLFT}
\label{sec:gqsp-nlft-equivalence}

\begin{table}
    \centering
    \renewcommand{\arraystretch}{1.5}
    \begin{tabular}{c|c|c|c}
        \toprule
        \textbf{Sequence} & \textbf{NLFT/QSP $(P, Q)$} & \textbf{QSP Operators} & \textbf{Phase factors} \\
        \midrule
        \begin{tabular}{c}
            $F_k \in i\R$ \\ \cite{alexisInfiniteQuantumSignal2024} \\
        \end{tabular} \vspace*{0.1em} &
        $P \in \R[z], Q \in i\R[z]$ &
        \begin{tabular}{c}
            $A_k = e^{i\phi_k X}$ \\
            $\phi_k \in (-\frac{\pi}{2}, \frac{\pi}{2})$
        \end{tabular} &
        $F_k = i \tan(\phi_k)$ \\
        \hline
        $F_k \in \R$ &
        $P \in \R[z], Q \in \R[z]$ &
        \begin{tabular}{c}
            $A_k = e^{i\phi_k Y}$ \\
            $\phi_k \in (-\frac{\pi}{2}, \frac{\pi}{2})$
        \end{tabular} & $F_k = \tan(\phi_k)$ \\
        \hline
        \begin{tabular}{c} $F_k \in \C$ \end{tabular} &
        $P \in \C[z], Q \in \C[z]$ &
        \begin{tabular}{c}
            $A_0 = e^{i\lambda Z} e^{i\phi_0 X} e^{i\theta_0 Z}$ \\
            $A_k = e^{i\phi_k X} e^{i\theta_k Z}$
        \end{tabular} &
        \begin{tabular}{c}
            $F_k = i \tan(\phi_k) e^{2i \psi_k}$ \\
            $\psi_k = \lambda + \sum_{j=0}^{k-1} \theta_j$
        \end{tabular} \\
        \bottomrule
    \end{tabular}
    \caption{Overview of the relationship between the NLFT and the variants of QSP, summarizing Section~\ref{sec:gqsp-nlft-equivalence}. These hold for both Laurent and analytic QSP. The processing operators for the general complex case are the $SU(2)$ formulation of the generalized QSP of~\cite{motlaghGeneralizedQuantumSignal2024}, as stated in Theorem~\ref{thm:generalized-qsp}.}
    \label{tab:qsp-nlft-relation}
\end{table}

As we will see in Section~\ref{sec:rhw-algorithm}, the Riemann-Hilbert-Weiss algorithm of~\cite{alexisInfiniteQuantumSignal2024}, as well as the version augmented with the Half Cholesky method of~\cite{niFastPhaseFactor2024}, allows to compute the inverse NLFT, even when $b$ is not achievable with $X$-constrained symmetric QSP.

The bottleneck here is Lemma 2 in~\cite{alexisQuantumSignalProcessing2024}, which shows that phase factors for $X$-constrained symmetric QSP can be computed the phase factors through the relation:
\begin{align}
    F_k = i \tan(\phi_k) \Longleftrightarrow \phi_k = \arctan(-i F_k) \label{eq:phase-factor-from-nlft-1}
\end{align}
where $F_k$ is guaranteed to be purely imaginary, by the symmetry constraints on $(a, b)$. Here we extend this result --- with a simpler, elementary argument --- showing that, while a purely imaginary $F_k$ yields the phase factors for an $X$-constrained QSP protocol, real sequences can similarly be translated to $Y$-constrained QSP protocols, and generally complex sequences are converted to Generalized QSP protocols (in the sense of Theorem~\ref{thm:generalized-qsp}). Table~\ref{tab:qsp-nlft-relation} summarizes the relationship between QSP and NLFT discussed in this section.

\begin{lemma}
    \label{thm:nlft-qsp-analytic}
    Let $F = (F_k)_{k \in \Z}$ be a Non-Linear Fourier sequence supported on $\{ 0, \ldots, n \}$, and assume the NLFT maps this sequence to a pair $(a, b) \in \ftH$. Then, by constructing the processing operators
    \begin{align*}
        A_k := \frac{1}{\sqrt{1 + |F_k|^2}}
        \begin{bmatrix}
            1 & F_k \\
            -F_k^* & 1
        \end{bmatrix}
    \end{align*}
    The analytic QSP protocol gives
    \begin{align*}
        A_0 \sigw A_1 \sigw \cdots \sigw A_n = \calG_F(z) \cdot \sigw^{n} =
        \begin{bmatrix}
            z^n a(z) & b(z) \\
            - z^n b^*(z) & a^*(z)
        \end{bmatrix}
    \end{align*}
\end{lemma}
\begin{proof}
    Notice that the term in the NLFT can be decomposed as
    \begin{align*}
        \frac{1}{\sqrt{1 + |F_k|^2}}
        \begin{bmatrix}
            1 & F_k z^k \\
            -F_k^* z^{-k} & 1
        \end{bmatrix}
        =
        \frac{1}{\sqrt{1 + |F_k|^2}}
        \begin{bmatrix}
            z^k & 0 \\ 
            0 & 1
        \end{bmatrix}
        \begin{bmatrix}
            1 & F_k \\
            -F_k^* & 1
        \end{bmatrix}
        \begin{bmatrix}
            z^{-k} & 0 \\ 
            0 & 1
        \end{bmatrix}
    \end{align*}
    By this decomposition, we can rewrite the NLFT as follows
    \begin{align*}
        \calG_F(z) = \prod_{k = 0}^n \sigw^k A_k \sigw^{-k}
    \end{align*}
    By shifting the occurrences of $\sigw$ in contiguous terms, we can see that between each $A_k$ we have $\sigw^{-k} \sigw^{k+1} = \sigw$. The product thus simplifies to
    \begin{align*}
        \calG_F(z) = \sigw^0 A_0 \sigw A_1 \sigw A_2 \sigw \cdots \sigw A_n \sigw^{-n}
    \end{align*}
    proving the claim.
\end{proof}
\noindent By the same argument, we can indeed prove that, if $F_n$ is supported on $\{ n_1, \ldots, n_2 \}$, then we obtain
\begin{align*}
    A_{n_1} \sigw A_{n_1+1} \sigw \cdots \sigw A_{n_2} = \sigw^{-n_1} \calG_F(z) \sigw^{n_2} = \begin{bmatrix}
        z^{n_2 - n_1} a(z) & z^{-n_1} b(z) \\
        - z^{n_2} b^*(z) & a^*(z)
    \end{bmatrix} \ .
\end{align*}
For the reader more familiar with the Laurent picture, we give the analogous claim, which directly follows from Lemma~\ref{thm:analytic-laurent-correspondence}.
\begin{corollary}
    \label{thm:nlft-qsp-laurent}
    Let $F = (F_k)_{k \in \Z}$, $(a(z), b(z))$ and $A_k$ as in Lemma~\ref{thm:nlft-qsp-analytic}. The Laurent QSP protocol yields
    \begin{align*}
        A_0 \sigv A_1 \sigv \cdots \sigv A_n = \calG_F(z^2) \cdot \sigv^{n} =
        \begin{bmatrix}
            z^n a(z^2) & z^{-n} b(z^2) \\
            - z^n b^*(z^2) & z^{-n} a^*(z^2)
        \end{bmatrix} \ .
    \end{align*}
\end{corollary}
We nonetheless hope to convey the message that working with the Laurent picture is quite confusionary in this context --- since Lemma~\ref{thm:nlft-qsp-analytic} is more immediate to understand ---, and it is always better to reduce to analytic QSP in order to work with the NLFT, while keeping in mind that returning to the Laurent formulation is free and always possible, thanks to Lemma~\ref{thm:analytic-laurent-correspondence}.

\subsection{Retrieving the GQSP phase factors from the NLFT sequence}
Suppose to have a protocol in the form of Lemma~\ref{thm:nlft-qsp-analytic}, obtained by computing the inverse NLFT. Such $A_k$ admits an Euler decomposition~\cite{nielsenQuantumComputationQuantum2010a}:
\begin{align}
    A_k = e^{i \psi_k Z} e^{i \phi_k X} e^{-i \psi_k Z} =
    \begin{bmatrix}
        \cos \phi_k & i e^{2i \psi_k} \sin \phi_k \\
        i e^{-2i \psi_k} \sin \phi_k & \cos \phi_k
    \end{bmatrix}
    \stackrel{!}{=}
    \frac{1}{\sqrt{1 + |F_k|^2}}
    \begin{bmatrix}
        1 & F_k \\
        -F_k^* & 1
    \end{bmatrix}
    \label{eq:processing-operator-euler-decomp}
\end{align}
which yields the relation $F_k = i \tan(\phi_k) e^{2i\psi_k}$. A solution $\psi_k \in [-\frac{\pi}{4}, \frac{\pi}{4}), \phi_k \in (-\frac{\pi}{2}, \frac{\pi}{2})$ is the following:
\begin{align*}
    \psi_k & = 
    \begin{cases}
        \frac{1}{2} \arctan\qty(\frac{\Im(-i F_k)}{\Re(-i F_k)}) = - \frac{1}{2} \arctan\qty(\frac{\Re(F_k)}{\Im(F_k)}) & F_k \not\in \R \\
        -\frac{\pi}{4} & F_k \in \R \setminus \{ 0 \} \\
        0 & F_k = 0
    \end{cases} \\
    \phi_k & = \arctan(-i e^{-2i \psi_k} F_k)
\end{align*}
We call $\psi_k$ \emph{phase prefactors}. This is not the only possible solution, as the sign of $F_k$ is controlled by both $\phi_k$ and $\psi_k$. Moreover, we have to cover the case $F_k \in \R$ separately: we choose $-\pi/4$ when $F_k$ is a non-zero real number, so that we obtain the processing operator
$$e^{-i\frac{\pi}{4}Z} e^{i\phi_k X} e^{i\frac{\pi}{4}Z} = e^{i\phi_k Y}$$
with $\phi_k = \arctan(F_k)$. When $F_k = 0$, we may choose any $\psi_k$ we want, but we need $\psi_k = 0$ in what follows. With these choices, the given solution is the only one that generalizes~(\ref{eq:phase-factor-from-nlft-1}). What we obtain is a protocol that applies three rotations at each step. To recover the GQSP protocol of Theorem~\ref{thm:generalized-qsp} we simply note that $Z$ rotations commute with $\sigw$ and $\sigv$.
\begin{align*}
    e^{-i \psi_k Z} \sigw\ e^{i \psi_{k+1} Z} = e^{i (\psi_{k+1} - \psi_k) Z} \sigw =: e^{i \theta_k Z} \sigw \ .
\end{align*}
This concludes the relation, which we summarize in
\begin{theorem}
    \label{thm:nlft-to-gqsp}
    Let $(F_k)_{k \in \Z}$ be a sequence supported in $\{ 0, \ldots, n \}$ whose NLFT is $(a(z), b(z))$. A GQSP protocol achieving $(P, Q) := (z^n a, b)$ is
    \begin{align*}
        e^{i\lambda Z} e^{i\phi_0 X} e^{i\theta_0 Z} \ \sigw\ e^{i\phi_1 X} e^{i\theta_1 Z} \ \sigw \cdots \sigw \ e^{i\phi_n X} e^{i\theta_n Z}
    \end{align*}
    where the generalized phase factors can be computed as
    \begin{align}
        \begin{cases}
            \lambda = \psi_0 \\
            \theta_k = \psi_{k+1} - \psi_k \\
            \phi_k = \arctan(-i e^{-2i \psi_k} F_k)
        \end{cases}
        \label{eq:gqsp-phase-factor-computation}
    \end{align}
    and $\psi_k$ being
    \begin{align*}
        \psi_k & = 
        \begin{cases}
            - \frac{1}{2} \arctan\qty(\frac{\Re(F_k)}{\Im(F_k)}) & F_k \not\in \R \\
            -\frac{\pi}{4} & F_k \in \R \setminus \{ 0 \} \\
            0 & F_k = 0
        \end{cases}
    \end{align*}
\end{theorem}
\noindent This result has a bunch of notable properties: first of all, the phase factors can be computed independently like in the case of~\cite{alexisQuantumSignalProcessing2024}, since $\lambda$ depends only on $F_0$, $\phi_k$ only on $F_k$, and $\theta_k$ only on $F_k, F_{k+1}$. This is why we chose $\psi_k = 0$ when $F_k = 0$, otherwise we would have non-zero phase factors even outside of the support of $F$, which would nevertheless resolve to the identity.

If $F$ is purely imaginary, then $\psi_k = 0$ for all $k$, and we reduce to~(\ref{eq:phase-factor-from-nlft-1}), implying that the GQSP protocol automatically turns into an $X$-constrained QSP. On the other hand, if $F$ is real, then the $A_k$ will turn out to be $Y$-rotations, and the GQSP protocol reduces to $Y$-constrained QSP. As a consequence, a NLFT algorithm can compute arbitrary GQSP phase factors, but whenever $(a, b)$ happens to have a certain symmetry property, then Lemma~\ref{thm:nlft-to-gqsp} will reduce the required rotation operations of the QSP protocol.

\subsection{Retrieving the NLFT sequence from a GQSP protocol}
For completeness, we also briefly discuss how to obtain a NLFT sequence from a GQSP protocol. Considering a protocol of the form as in Theorem~\ref{thm:generalized-qsp}:
\begin{align*}
    e^{i\lambda Z} e^{i\phi_0 X} e^{i\theta_0 Z} \ \sigw\ e^{i\phi_1 X} e^{i\theta_1 Z} \ \sigw \cdots \sigw \ e^{i\phi_n X} e^{i\theta_n Z} \ket{0} = (P, Q)
\end{align*}
We assume without loss of generality that the leading coefficient of $P$ is real and positive: otherwise, if this coefficient has phase $e^{2i\alpha}$, we can replace $\lambda \leftarrow \lambda - \alpha, \theta_n \leftarrow \theta_n - \alpha$ and the new protocol will generate $(P e^{-2i\alpha}, Q)$. Since $(a, b) := (z^{-n} P, Q) \in \ftH$, the $A_k$ can be rewritten in the form $(\ref{eq:processing-operator-euler-decomp})$ by Lemma~\ref{thm:nlft-qsp-analytic}, thus implying the condition on the phase factors:
\begin{align}
    \lambda + \sum_{k = 0}^n \theta_k = 0 \label{eq:canonical-condition}
\end{align}
We call GQSP protocols (or, equivalently, sets of phase factors) satisfying this condition \emph{canonical}, which constitute a uniqueness result for GQSP, thanks to the bijectivity of the NLFT:
\begin{corollary}
    \label{thm:gqsp-uniqueness}
    Let $P, Q \in \C[z]$ of degree $n$, where the leading coefficient of $P$ is real and positive. There exists a unique canonical set of phase factors $\lambda, \phi_k \in [-\frac{\pi}{2}, \frac{\pi}{2}), \theta_k \in [-\pi, \pi)$ such that:
    \begin{align*}
        e^{i\lambda Z} e^{i\phi_0 X} e^{i\theta_0 Z} \ \sigw\ e^{i\phi_1 X} e^{i\theta_1 Z} \ \sigw \cdots \sigw \ e^{i\phi_n X} e^{i \theta_n Z} \ket{0} = (P, Q)
    \end{align*}
    if and only if $|P(z)|^2 + |Q(z)|^2 \equiv 1$ on $\T$.
\end{corollary}
\noindent In order to retrieve the inverse NLFT $F_k$ of $(z^{-n} P, Q)$, we solve the linear system obtained by inverting (\ref{eq:gqsp-phase-factor-computation}):
\begin{align*}
    \begin{cases}
        \psi_0 = \lambda \\
        \psi_{k+1} - \psi_k = \theta_k & \text{for } 0 \le k < n
    \end{cases}
\end{align*}
This $(n+1) \times (n+1)$ linear system has the unique solution $\psi_{k+1} = \lambda + \sum_{j = 0}^{k-1} \theta_j$, and because of Eq.~(\ref{eq:canonical-condition}) we also have $-\psi_n = \theta_n$. The new set of phase factors can then be used to construct a protocol
\begin{align*}
    A_0 \sigw A_1 \sigw \cdots \sigw A_n
\end{align*}
with $A_k = e^{i \psi_k Z} e^{i \phi_k X} e^{-i \psi_k Z}$, and $F_k = i \tan(\phi_k) e^{2i\psi_k}$, as already found in (\ref{eq:processing-operator-euler-decomp}).

\subsection{Switching the polynomials}
While the convention with the NLFT is to construct $(a, b)$, for a desired $b$, the current convention in QSP the target polynomial is actually the left polynomial $P$ in $(P, Q)$. Once we found the inverse NLFT, this section showed so far how to find phase factors for $(Q, P)$. Finding the phase factors for $(P, Q)$ is straightforward by multiplying on the right by $iX$.
\begin{align*}
    (Q, P) \cdot iX = (iP, iQ)
\end{align*}
In order to obtain the desired $P$, we transform $P \rightarrow -iP$ (or equivalently, we multiply the obtain inverse NLFT by $-i$). The $iX$ can be easily incorporated in the last GQSP processing operator using the properties of Pauli matrices:
\begin{align*}
    e^{i\phi_n X} e^{i\theta_n Z} \cdot iX = e^{i(\phi_n + \frac{\pi}{2}) X} e^{-i\theta_n Z} \ .
\end{align*}
Thus, it is sufficient to transform $\phi_n \leftarrow \phi_n + \frac{\pi}{2}, \theta_n \leftarrow -\theta_n$ to obtain $(P, iQ)$. We multiply by $iX$ instead of appearently more convenient $X$, because $iX$ preserves the subalgebra of $X$-constrained QSP. Indeed, if $P$ turns out to be real, then $-iP$ will be purely imaginary, and the complementary $Q$ will be real. Therefore, the phase factors for $(P, iQ)$ will be $X$-constrained. Similarly, one can swap the two polynomials in the $Y$-constrained subalgebra by multiplying by $iY$ on the right.
\section{Review of the Riemann-Hilbert-Weiss algorithm}
\label{sec:rhw-algorithm}
The Riemann-Hilbert-Weiss algorithm proposed in~\cite{alexisInfiniteQuantumSignal2024} takes some polynomial $b(z)$, and finds (1) a complementary polynomial $a(z)$ such that $(a, b) \in \ftH$, and (2) the sequence $F_k$ whose NLFT is $(a, b)$. This algorithm retrieves the $F_k$'s (up to precision $\epsilon$) of some $b$ with $\norm{b}_{\T} \le 1 - \eta$ supported on $\{ 0, \ldots, n \}$\footnote{Unlike~\cite{alexisInfiniteQuantumSignal2024}, we will always assume the support starts at zero, i.e., $b \in H^2(\D)$. This is without loss of generality by the shift property of Theorem~\ref{thm:nlft-properties}.} in $\bigO(n^4 + n/\eta \log \epsilon)$ time, but the Half Cholesky method of~\cite{niFastPhaseFactor2024} (revised in Section~\ref{sec:rhw-half-cholesky}) brings down the complexity to $\bigO(n^2 + n/\eta \log \epsilon)$.

While the authors in~\cite{alexisInfiniteQuantumSignal2024} prove the correctness of the algorithm in the case of $b$ having only imaginary coefficients, we show that little modification allows correct retrieval of $F_k$ for generally complex function $b$, without any penalty in complexity. The relationship with the phase factors shown in Section~\ref{sec:gqsp-nlft-equivalence} will then derive a corresponding GQSP protocol. It is worth remarking that the same algorithm is able to find phase factors for \emph{any} of the QSP variants we have shown --- from reflection to analytic --- and, whenever our polynomial $b$ falls into a subalgebra (e.g., it has real or imaginary coefficients), the returned phase factors will naturally fall into the corresponding subspace (e.g., they'll become $Y/X-$rotations, respectively). The reader can find a Python package~\cite{rhwPythonRepo} implementing the unified algorithm.

\begin{algorithm}
    \caption{Weiss' algorithm~\cite{alexisInfiniteQuantumSignal2024}}\label{alg:weiss}
    \begin{algorithmic}
        \State \textbf{Input:} A polynomial $b(z) = \sum_{k = 0}^n b_k z^k$ with $\lVert b \rVert_\infty \le 1 - \eta$, and $\epsilon > 0$.
        \State \textbf{Output:} The Fourier coefficients $(\hat{c}_0, \ldots, \hat{c}_n)$ of $b/a$ in $\{ 0, \ldots, n \}$, up to precision $\epsilon$.
        \State Let $N$ be the smallest power of two with $N \ge \frac{8n}{\eta} \log\qty(\frac{576 n^2}{\eta^4 \epsilon})$.
        \State $\hat{r} = (\hat{r}_{-N/2}, \ldots, \hat{r}_{N/2 - 1}) \leftarrow $ Fourier approximation of $R(z) = \log\sqrt{1 - \abs{b(z)}^2}$
        \State $\hat{G}^*(z) := \hat{r}_0 + 2 \sum_{k = 1}^{N/2} r_{-k} z^{-k}$ \Comment{Schwarz integral of $R(z)$}
        \State $\hat{c} = (\hat{c}_{-N/2}, \ldots, \hat{c}_{N/2 - 1}) \leftarrow$ Fourier approximation of $b(z) e^{-\hat{G}^*(z)}$ \\
        \Return $(\hat{c}_0, \ldots, \hat{c}_n)$
    \end{algorithmic}
\end{algorithm}

\subsection{Weiss' algorithm and the completion problem}
The first procedure --- which we recall in Algorithm~\ref{alg:weiss} --- is Weiss' algorithm.
The idea is quite simple, and we explain it using elements from~\cite{berntsonComplementaryPolynomialsQuantum2025}: the function $$f(z) = \sqrt{1 - \abs{b(z)}^2}$$ trivially satisfies $\abs{f(z)}^2 + \abs{b(z)}^2 = 1$, but it is not a polynomial. We take $R(z) := \log f(z)$ and our question reduces to finding a function $G(z)$ such that (1) is analytic in $\D$ (i.e., it does not have negative-degree terms in its Laurent series expansion) and $R(z) = \Re{G(z)}$\footnote{The imaginary part of $G$ is its \emph{Hilbert transform}, as shown in~\cite{alexisInfiniteQuantumSignal2024}.}. This implies that
\begin{align}
    a(z) := e^{G^*(z)}
    \label{eq:a-outer-function}
\end{align}
is analytic on $\D^*$, but still satisfies $|a|^2 + |b|^2 = 1$. Such $G(z)$ can be found using the Schwarz integral formula~\cite{berntsonComplementaryPolynomialsQuantum2025}:
\begin{align*}
    G(z) = \dashint_{\T} R(\zeta) \frac{\zeta + z}{\zeta - z} \frac{\dd \zeta}{\zeta} \ .
\end{align*}
Since $a$ is analytic on $\D^*$ and $a a^* = 1 - b b^*$ is a Laurent polynomial, then $a$ is a polynomial as well. The advantage of the Schwarz kernel is that it acts very simply on the Laurent series
\begin{align*}
    \dashint_{\T} \zeta^k \frac{\zeta + z}{\zeta - z} \frac{\dd \zeta}{\zeta} =
    \begin{cases}
        2 z^k & k > 0 \\
        1 & k = 0 \\
        0 & k < 0
    \end{cases}
\end{align*}
Therefore, applying the Schwarz transform on the Laurent series of $R(z)$ is straightforward. Weiss' algorithm uses the Fast Fourier Transform (on a sufficiently large group order) to compute a Fourier approximation of $R(z)$, it applies the Schwarz transform on its coefficients, and then computes the Fourier approximation of $e^{G^*(z)}$. While this algorithm can be used to compute $a$, we want to use it to compute an approximation of the ratio $b/a$, for the next part of the algorithm (which we can readily compute as the Fourier expansion of $b(z) e^{-G^*(z)}$). No substantial change is made to this algorithm, except for the fact that we do not assume the coefficients of $b(z)$ to be imaginary, nor that $b(z) = b(z^{-1})$. This implies that the $\hat{c}_k$ are not guaranteed to be imaginary, but this does not hinder the correctness of the algorithm, nor does it influence the error and complexity analysis given in~\cite{alexisInfiniteQuantumSignal2024}.

\subsection{Riemann-Hilbert algorithm}
In the second part of the algorithm, one computes $F_0$ by slicing the pair of polynomials
\begin{align*}
    (a, b) = (a_-, b_-) (a_+, b_+)
\end{align*}
where $b_-$ is supported only on negative frequencies, and $b_+$ is supported only on non-negative ones. This goes by the name of \emph{Riemann-Hilbert factorization}. Once one has $(a_+, b_+)$, then by the composition property $F_0$ can be computed as
\begin{align*}
    F_0 = \frac{b_+(0)}{a_+^*(0)}
\end{align*}
Analogously, one computes $F_k$ by carrying out a Riemann-Hilbert factorization on $(a, z^{-k} b)$, by leveraging the shift property. The only question is how to find $(a_+, b_+)$ from $(a, b)$, and it was proven in~\cite{alexisInfiniteQuantumSignal2024} that $a_+, b_+$ can be retrieved (up to a common multiplicative constant) as the unique solution of the linear system:
\begin{align*}
    \begin{bmatrix}
        \id & P_{\D^*} \frac{b^*}{a^*} \\
        - P_{\D} \frac{b}{a} & \id
    \end{bmatrix}
    \begin{bmatrix}
        A_+ \\ B_+
    \end{bmatrix}
    =
    \begin{bmatrix}
        1 \\ 0
    \end{bmatrix}
\end{align*}
where $P_{\D}, P_{\D^*}$ are the Cauchy projections on the Hardy spaces $H^2(\D), H^2(\D^*)$, respectively. One can see that this matrix can be written as $\id + M$, where $M$ is an anti-Hermitian operator, which guarantees the full operator to be invertible. This is true regardless of whether $b$ is imaginary, and thus a Riemann-Hilbert factorization always exists\footnote{This exists and can be obtained with this argument even if $b$ only satisfies the Szeg{\"o} condition~\cite{alexisInfiniteQuantumSignal2024}, i.e., $\int_\T \log(1 - |b|^2) > -\infty$.}. It is also proven that, provided $b$ is purely imaginary, such linear system is equivalent to the following finite linear system (following~\cite{niFastPhaseFactor2024}):
\begin{align*}
    \begin{bmatrix}
        \id & -T^T \\
        -T & \id
    \end{bmatrix}
    \begin{bmatrix}
        \vv{b} \\ \text{rev}(\vv{a})
    \end{bmatrix}
    =
    \begin{bmatrix}
        \vv{0} \\ \text{rev}(\vv{e}_0)
    \end{bmatrix}
\end{align*}
where $\vv{e}_0$ is the element of the standard basis, $\text{rev}(\vv{v})$ reverses the order of the elements of $\vv{v}$, and $T$ is the Toeplitz matrix having $(c_n, \ldots, c_0)^T$ as first column. This can be extended to the general case by slightly modifying the system:
\begin{align*}
    \begin{bmatrix}
        \id & -T^T \\
        T^* & \id
    \end{bmatrix}
    \begin{bmatrix}
        \vv{b} \\ \text{rev}(\vv{a})
    \end{bmatrix}
    =
    \begin{bmatrix}
        \vv{0} \\ \text{rev}(\vv{e}_0)
    \end{bmatrix}
\end{align*}
where $T^*$ is the matrix obtained by conjugating $T$ entry-wise. In the case of $b$ being imaginary, also $c$ is imaginary, and $T$ has imaginary entries as well, making $T^* = -T$, and recovering the base case. The Riemann-Hilbert-Weiss procedure is summarized in Algorithm~\ref{alg:rhw}.
\begin{algorithm}
    \caption{Generalized Riemann-Hilbert-Weiss algorithm~\cite{alexisInfiniteQuantumSignal2024}}\label{alg:rhw}
    \begin{algorithmic}
        \State \textbf{Input:} A polynomial $b(z) = \sum_{k = 0}^n b_k z^k$ with $\lVert b \rVert_\infty \le 1 - \eta$, and $\epsilon > 0$.
        \State \textbf{Output:} A sequence $F = (F_0, \ldots, F_n)$ such that $\calG_F(z) = (\cdot, b)$, up to precision $\epsilon$.
        \State $\hat{c} = (\hat{c}_{-N/2}, \ldots, \hat{c}_{N/2 - 1}) \leftarrow$ \textsc{Weiss}(b) \Comment{See Algorithm~\ref{alg:weiss}}
        \For{$k \in \{0, \ldots, n\}$}
            \State $T_k \leftarrow$ Toeplitz matrix with $(c_n, \ldots, c_k)^T$ as first column.
            \State Solve the system
            \begin{align*}
                \begin{bmatrix}
                    \id & -T_k^T \\
                    T^*_k & \id
                \end{bmatrix}
                \begin{bmatrix}
                    \vv{b}_k \\ \text{rev}(\vv{a}_k)
                \end{bmatrix}
                =
                \begin{bmatrix}
                    \vv{0} \\ \text{rev}(\vv{e}_0)
                \end{bmatrix}
            \end{align*}
            \State $F_k \leftarrow b_{k,0}/a_{k,0}$
        \EndFor \\
        \Return $F_0, \ldots, F_n$
    \end{algorithmic}
\end{algorithm}

\subsection{Half-Cholesky method}
\label{sec:rhw-half-cholesky}
Algorithm~\ref{alg:rhw} solves a linear system to compute each $F_k$, but it is possible to do better. In~\cite{niFastPhaseFactor2024}, the authors show that the system matrix for $F_k$ is a submatrix of the one for $F_{k-1}$, which allows to solve all the systems at once with a $LDL$ decomposition. We review their argument, showing that much like the imaginary case, also the general case can be solved similarly by considering general complex matrices.
\begin{align*}
    \begin{bmatrix}
        \id & -T^T \\
        T^* & \id
    \end{bmatrix}
    \begin{bNiceArray}{c|cccc}[margin=6pt]
        \Block{4-1}<\Large>{\id}
        &b_{n,0} & b_{n-1,0}&\cdots & b_{0,0}\\
        & & b_{n-1,1}&&b_{0,1}\\
        &&&\ddots&\vdots\\
        &&&& b_{0,n}\\
        \hline
        &a_{n,0} & a_{n-1,1}&\cdots & a_{0,n}\\
        & & a_{n-1,0}&&a_{0,n-1}\\
        &&&\ddots&\vdots\\
        &&&& a_{0,0}
    \end{bNiceArray}
    =
    \begin{bNiceArray}{c|cccc}[margin=6pt]
        \Block{4-1}<\Large>{\id} 
        &&&&\\
        &&&&\\
        &&&&\\
        &&&&\\
        \hline
        \Block{4-1}<\Large>{T^*}
        &1 & & & \\
        &* & 1&&\\
        &\vdots&&\ddots&\\
        &*&*&\cdots& 1
    \end{bNiceArray}
\end{align*}
The matrix in the middle is invertible, since the constant coefficients of $a$ are always real and positive. Thus if we call $A$ the matrix on the left, a $LU$ decomposition $A = L_A U_A$ would give all the $b_{k,0}, a_{k,0}$ in $U_A^{-1}$. This method would already bring down the complexity of the Riemann-Hilbert algorithm from $\Tilde{\bigO}(n^4)$ to $\Tilde{\bigO}(n^3)$, but we can do better. By calling $B = T^*$, we aim for a $LDL$ decomposition:
\begin{align*}
    A =
    \begin{bmatrix}
        \id & -B^\dag \\
        B & \id
    \end{bmatrix}
    & =
    \begin{bmatrix}
        \id & \\
        B & \id
    \end{bmatrix}
    \begin{bmatrix}
        \id & \\
        & \id + B B^\dag
    \end{bmatrix}
    \begin{bmatrix}
        \id & -B^\dag \\
        & \id
    \end{bmatrix} \\
    & =
    \begin{bmatrix}
        \id & \\
        B & L
    \end{bmatrix}
    \begin{bmatrix}
        \id & \\
        & D L^\dag
    \end{bmatrix}
    \begin{bmatrix}
        \id & -B^\dag \\
        & \id
    \end{bmatrix} \ .
\end{align*}
The second equality is taken by decomposing $\id + BB^\dag = LDL^\dag$ (this is possible as the left-hand side is clearly positive definite). The product of the two matrices on the right is exactly $U_A$, which we can invert using Schur's complement formula.
\begin{align*}
    U_A^{-1} =
    \begin{bmatrix}
        \id & B^\dag \\
        & \id
    \end{bmatrix}
    \begin{bmatrix}
        \id & \\
        & L^{-\dag} D^{-1}
    \end{bmatrix}
    =
    \begin{bmatrix}
        \id & B^\dag (DL^\dag)^{-1} \\
        0 & (DL^\dag)^{-1}
    \end{bmatrix}
\end{align*}
Since we are only interested in the $b_{k,0}, a_{k,0}$, we want to extract only the first row of $B^{\dag} (DL^\dag)^{-1}$ and the diagonal of $(DL^\dag)^{-1}$ (which is simply the diagonal of $D^{-1}$, since $L$ has ones on its diagonal).
\begin{align*}
    (1, 0, \ldots, 0) B^\dag (D L^\dag)^{-1} & = (b_{n,0}, \ldots, b_{0,0}) \\
    D & = \diag(1/a_{n,0}, \ldots, 1/a_{0,0})
\end{align*}
Thus our final solution can be computed as:
\begin{align*}
    \begin{bmatrix}
        F^*_n \\ \vdots \\ F^*_0
    \end{bmatrix}
    =
    \begin{bmatrix}
        b^*_{n,0}/a_{n,0} \\ \vdots \\ b^*_{0,0}/a_{0,0}
    \end{bmatrix}
    = D D^{-1} L^{-1} B
    \begin{bmatrix}
        1 \\ \vdots \\ 0
    \end{bmatrix}
    = L^{-1}
    \begin{bmatrix}
        c_n^* \\ \vdots \\ c_0^*
    \end{bmatrix}
    =: L^{-1} \vv{p}
\end{align*}
The only difference with~\cite{niFastPhaseFactor2024} is that the decompositions are done over the complex field, and thus one must pay attention to the conjugation of the entries in $B, L$. By computing $L$ in the $LDL$ decomposition of $K := \id + BB^\dag$ yields the sequence. This would still take $\Tilde{\bigO}(n^3)$, but the authors in~\cite{niFastPhaseFactor2024} take advantage of the \emph{displacement structure}~\cite{sayedFastAlgorithmsGeneralized1995}:
\begin{align*}
    K - Z K Z^\dag = G G^\dag
\end{align*}
where $Z$ is the lower shift matrix and $G = [\vv{e}_0, \vv{p}] \in \C^{(n+1) \times 2}$. This additional information shows that $K$ is fully determined by $G$, and the $\bigO(n^2)$ algorithm retrieving $L$ shown in~\cite{niFastPhaseFactor2024} works also in the general case, up to a switch from real to complex field.

\subsection{Numerical demonstration}
We show here the numerical performance of the generalized Riemann-Hilbert-Weiss algorithm and compare it with two previous methods: the convolution optimization method~\cite{motlaghGeneralizedQuantumSignal2024} and Prony's method~\cite{yamamotoRobustAngleFinding2024}, both conjugated with the plain layer stripping algorithm\footnote{The authors in~\cite{yamamotoRobustAngleFinding2024} call this \emph{carving}, while \emph{layer stripping} is a term that comes from the NLFT literature.}. All the methods are given the same random polynomial with $\eta = 0.5$. Both previous methods give a comparable runtime with the Half-Cholesky solver, but Prony's method failed to give an accurate solution after $n \simeq 100$, while the convolution optimization algorithm consistently gave an error between $10^{-3}$ and $10^{-4}$, which is likely not enough to survive the stripping procedure\footnote{Recently, the authors in~\cite{niInverseNonlinearFast2025} proved that the layer stripping algorithm is stable when the complementary polynomial is outer, but showed non-outer examples when stability is lost.}. Figure~\ref{fig:numerical-comparison} summarizes the results.

\begin{figure}
    \begin{center}
        \begin{subfigure}[t]{0.49\textwidth}
            \centering
            \begin{tikzpicture}
    \begin{loglogaxis}[
        xlabel={Degree},
        title={Time (s)},
        legend style={at={(axis cs:12700,1)},anchor=south west},
        grid=both,
        width=8cm,
        height=5.5cm,
        legend entries={},
    ]

    \addplot+[mark=*, color=blue] coordinates {
        (5, 0.604099)
        (10, 0.107597)
        (20, 0.055522)
        (50, 0.092187)
        (100, 0.240685)
        (200, 0.339805)
        (500, 0.805176)
        (1000, 2.178746)
        (2000, 8.944993)
        (5000, 42.772226)
        (10000, 163.382205)
    };

    \addplot+[mark=square*, color=red] coordinates {
        (5, 0.000774)
        (10, 0.000778)
        (20, 0.001507)
        (50, 0.006277)
        (100, 0.02227)
        (200, 0.08142)
        (500, 0.533492)
        (1000, 2.38536)
        (2000, 16.255877)
        (5000, 194.5627)
        (10000, 1527.635769)
    };

    \addplot+[mark=triangle*, color=darkgreen] coordinates {
        (5, 0.001228)
        (10, 0.002097)
        (20, 0.006187)
        (50, 0.055005)
        (100, 0.39201)
        (200, 3.078557)
        (500, 51.578559)
    };

    \addplot+[mark=diamond*, color=darkyellow] coordinates {
        (5, 0.001169)
        (10, 0.00161)
        (20, 0.002628)
        (50, 0.005533)
        (100, 0.011838)
        (200, 0.028156)
        (500, 0.107665)
        (1000, 0.360379)
        (2000, 1.542592)
        (5000, 10.129023)
        (10000, 54.537072)
    };

    \end{loglogaxis}
\end{tikzpicture}
        \end{subfigure}
        \hfill
        \begin{subfigure}[t]{0.49\textwidth}
            \centering
            \vspace*{-11em}
            \begin{tikzpicture}
    \begin{axis}[
        hide axis,
        xmin=0, xmax=1,
        ymin=0, ymax=1,
        legend columns=1,
        legend style={
            draw=none,
            font=\small,
            cells={anchor=west},
            anchor=north west,
            column sep=1em,
        }
    ]
    \addlegendimage{mark=*, color=blue}
    \addlegendentry{Minimize convolution + layer strip~\cite{motlaghGeneralizedQuantumSignal2024}}
    \addlegendimage{mark=square*, color=red}
    \addlegendentry{Prony's method + layer strip~\cite{yamamotoRobustAngleFinding2024}}
    \addlegendimage{mark=triangle*, color=darkgreen}
    \addlegendentry{Generalized RHW}
    \addlegendimage{mark=diamond*, color=darkyellow}
    \addlegendentry{Generalized RHW + Half-Cholesky}
    \end{axis}
\end{tikzpicture}
        \end{subfigure}
        \\
        \begin{subfigure}[t]{0.49\textwidth}
            \centering
            \begin{tikzpicture}
    \begin{loglogaxis}[
        xlabel={Degree},
        title={Completion Forward Error},
        grid=both,
        width=8cm,
        height=6cm,
    ]
    
    \addplot+[mark=*, color=blue] coordinates {
        (5, 3.9058948828090076e-05)
        (10, 0.0009732384983898672)
        (20, 0.0019137534740062885)
        (50, 0.0017380506669568293)
        (100, 0.0007965906472860171)
        (200, 0.0006214630944606345)
        (500, 0.0004478001978560943)
        (1000, 0.0004285694258335868)
        (2000, 0.0004162558693928458)
        (5000, 0.00040100544142523904)
        (10000, 0.00041245494242610296)
    };

    \addplot+[mark=square*, color=red] coordinates {
        (5, 9.22732118031501e-15)
        (10, 8.789069718231324e-14)
        (20, 1.6764111552121918e-12)
        (50, 0.4513945169813889)
        (100, 0.002289005587915917)
        (200, 14.766884871346768)
        (500, 16.45512860475663)
        (1000, 14.938720134824697)
        (2000, 12.5899898951599)
        (5000, 20.003652768983127)
        (10000, 73.73285804186976)
    };

    \addplot+[mark=triangle*, color=darkgreen] coordinates {
        (5, 2.4550811316481624e-16)
        (10, 1.2562398856048086e-16)
        (20, 1.8311014654306592e-16)
        (50, 1.660970955577654e-16)
        (100, 2.961619280638464e-16)
        (200, 2.0512540005190402e-16)
        (500, 3.398884215310824e-16)
    };

    \addplot+[mark=diamond*, color=darkyellow] coordinates {
        (5, 2.4550811316481624e-16)
        (10, 1.2562398856048086e-16)
        (20, 1.8311014654306592e-16)
        (50, 1.660970955577654e-16)
        (100, 2.961619280638464e-16)
        (200, 2.0512540005190402e-16)
        (500, 3.398884215310824e-16)
        (1000, 2.711105819969782e-16)
        (2000, 3.7070384894585946e-16)
        (5000, 3.368196022710338e-16)
        (10000, 3.490570391413021e-16)
    };
    \end{loglogaxis}
\end{tikzpicture}
        \end{subfigure}
        \hfill
        \begin{subfigure}[t]{0.49\textwidth}
            \centering
            \begin{tikzpicture}
    \begin{loglogaxis}[
        xlabel={Degree},
        title={Inverse NLFT forward error},
        grid=both,
        width=8cm,
        height=6cm,
    ]
    
    \addplot+[mark=*, color=blue] coordinates {
        (5, 0.0007013048405174351)
        (10, 0.0708783003221537)
        (20, 0.05814629533893324)
        (50, 0.9135690821373526)
        (100, 0.9277136528011141)
        (200, 1.0222987265675056)
        (500, 1.3767807175468454)
        (1000, 1.3217954182378917)
        (2000, 1.4645152777161343)
        (5000, 1.5645517093864854)
        (10000, 1.6374284746996228)
    };
    
    \addplot+[mark=square*, color=red] coordinates {
        (5, 7.241850991377079e-15)
        (10, 6.924825881018361e-14)
        (20, 1.272828170167208e-12)
        (50, 0.29070614971881037)
        (100, 0.001690801462474672)
        (200, 2.211139659189373)
        (500, 2.063506826168354)
        (1000, 2.178554107217126)
        (2000, 2.239040097129427)
        (5000, 2.882993125633712)
        (10000, 5.3591877669822185)
    };
    
    \addplot+[mark=triangle*, color=darkgreen] coordinates {
        (5, 3.2773521868536315e-16)
        (10, 1.8023423793820858e-16)
        (20, 6.968393310813226e-16)
        (50, 7.348166773229513e-16)
        (100, 9.287449232689912e-16)
        (200, 3.356754920353195e-15)
        (500, 1.2281266214061853e-14)
    };
    
    \addplot+[mark=diamond*, color=darkyellow] coordinates {
        (5, 3.145140370551792e-16)
        (10, 1.7529185531358286e-16)
        (20, 6.914050156849249e-16)
        (50, 7.224152096092452e-16)
        (100, 9.312294607735957e-16)
        (200, 3.544703160646856e-15)
        (500, 1.2280313568249785e-14)
        (1000, 2.1847148285461807e-14)
        (2000, 4.568428251118783e-14)
        (5000, 8.899294531366063e-14)
        (10000, 4.2960853481452844e-13)
    };

    \end{loglogaxis}
\end{tikzpicture}
        \end{subfigure}
    \end{center}
    \caption{Comparison of the generalized Riemann-Hilbert-Weiss algorithm (with and without Half-Cholesky) with previous methods, namely the convolution optimization + layer stripping of~\cite{motlaghGeneralizedQuantumSignal2024} and Prony's method + layer stripping~\cite{yamamotoRobustAngleFinding2024}. The completion error is computed as $\norm{|a(z)|^2 + |b(z)|^2 - 1}_{L^2}$, while the inverse NLFT error is the distance (again in $L^2$ norm) between the original pair $(a, b)$ and the reconstructed pair obtained from the NLFT of the computed sequence $F_k$.}
    \label{fig:numerical-comparison}
\end{figure}
\section{Discussion}
In this work we consider Non-Linear Fourier Analysis and the $SU(2)$-Fourier transform for different variants of quantum signal processing, extending the results of~\cite{alexisQuantumSignalProcessing2024,alexisInfiniteQuantumSignal2024}. In particular, we show that different variants of QSP (that correspond to subalgebras of the $SU(2)$ matrix polynomials~\cite{chaoFindingAnglesQuantum2020}), correspond to constraints on the NLFT sequence. In particular, a generally complex sequence can always be translated to a GQSP protocol.

We highlight that, in light of this work, a proof of Ya-Ju Tsai from 2005~\cite{tsaiSU2NonlinearFourier2005} characterizing the image of compactly supported sequences through the $SU(2)$-NLFT can be seen as an alternative proof for GQSP. The same statement gave a uniqueness result for the phase factors in GQSP, as shown in Corollary~\ref{thm:gqsp-uniqueness}. Moreover, $\ftH$ gives the image of a possible extension of GQSP to one-sided infinite sequences of phase factors. By the relation $F_k = i \tan(\phi_k) e^{2i\psi_k}$ found in~(\ref{eq:processing-operator-euler-decomp}), we can conclude $\phi_k \rightarrow 0$ (like in~\cite{dongInfiniteQuantumSignal2024}), but the prefactors $\psi_k$ need not converge, as they only control the phase of $F_k$. For the same reason, however, they do not influence the convergence of the processing operators, since
$$\norm{A_k - \id} = \norm{e^{i\psi_k Z} e^{i\phi_k X} e^{-i\psi_k Z} - \id} = 2 \sin(\frac{\abs{\phi_k}}{2})$$
which converges to $0$ regardless of the behaviour of $\psi_k$. In other words, the effects of $\theta_k$ in an infinite GQSP protocol become negligible as $k \rightarrow \infty$, as expected.

We also showed how the Riemann-Hilbert-Weiss algorithm in~\cite{alexisInfiniteQuantumSignal2024,niFastPhaseFactor2024} can be easily extended to generally complex sequences, giving a provably stable algorithm for the phase factors computation of GQSP, with the notable property that the produced protocol will be in a minimal subalgebra, i.e., whenever the desired polynomial pair $(P, Q)$ can be constructed with only $X/Y$ rotations, then the double GQSP rotations will automatically reduce to a single $X/Y$ rotation.

NLFA is a promising tool: more recent work managed to bring down the $\bigO(n^2)$ complexity of the Half-Cholesky method to $\bigO(n \log^2 n)$~\cite{niInverseNonlinearFast2025}. The complexification of QSP in~\cite{bastidasComplexificationQuantumSignal2024} shows how the theories of $SU(2)$-NLFT and $SU(1,1)$-NLFT are related to a more general non-unitary QSP. Moreover, NLFA might be a crucial tool in understanding the properties of M-QSP: in particular, the rigidity of the layer stripping argument --- which is not always possible in the multivariate case~\cite{laneveMultivariatePolynomialsAchievable2025,rossiMultivariableQuantumSignal2022,moriCommentMultivariableQuantum2024,itoPolynomialTimeConstructive2024} --- makes it hard to deduce characterizations of the M-QSP polynomials, and a treatment based on Riemann-Hilbert factorization is subject to future work.

\section*{Acknowledgments}
I would like to thank Michel Alexis for his invaluable help with Non-Linear Fourier Analysis, and Danial Motlagh for clarifications on GQSP. This work is supported by the Swiss National Science Foundation (SNSF), project No. 200020-214808.

\bibliographystyle{quantum}
\bibliography{refs,extras}

\end{document}